\documentclass[letterpaper]{article} 
\usepackage{aaai2026}  
\usepackage{times}  
\usepackage{helvet}  
\usepackage{courier}  
\usepackage[hyphens]{url}  
\usepackage{graphicx} 
\usepackage{natbib}  
\usepackage{caption} 
\usepackage{algorithm}

\usepackage{newfloat}
\usepackage{listings}
\DeclareCaptionStyle{ruled}{labelfont=normalfont,labelsep=colon,strut=off} 
\floatstyle{ruled}
\newfloat{listing}{tb}{lst}{}
\floatname{listing}{Listing}
\usepackage{latexsym}

\usepackage{multirow}%
\usepackage{amssymb,amsfonts}%
\usepackage{dsfont}
\usepackage{mathrsfs}%
\usepackage{xcolor}%
\usepackage{textcomp}%
\usepackage{manyfoot}%
\usepackage{algorithmicx}%
\usepackage{algpseudocode}%
\usepackage{amsmath}
\usepackage{verbatim}
\usepackage{booktabs}
\usepackage{tabularx}
\usepackage{makecell}
\usepackage{ragged2e}
\usepackage{enumitem} 
\usepackage{adjustbox}
\usepackage{lmodern}
\usepackage{caption}
\usepackage{subcaption}

\usepackage{longtable}
\usepackage{xspace}

\newtheorem{definition}{Definition}
\newtheorem{proposition}{Proposition}

\newcommand\restr[2]{{
  \left.\kern-\nulldelimiterspace 
  #1 
  \littletaller 
  \right|_{#2} 
  }}

\newcommand{\R}{\mathbb{R}}

\newcommand{\Rv}{{\pmb{R}}}
\newcommand{\A}{\mathcal{A}}
\newcommand{\T}{\mathcal{T}}

\newcommand{\abs}[1]{{\left|#1\right|}}

\DeclareMathOperator*{\argmin}{arg\,min}

\newcommand{\VS}{\text{VS}\xspace}
\newcommand{\DS}{\text{DS}\xspace}
\newcommand{\loss}{\mathcal{L}}
\newcommand{\Ultra}{\textsc{UltraFeedback}\xspace}
\newcommand{\PKU}{\textsc{PKU-Align-Anything}\xspace}
\newcolumntype{Y}{>{\Centering\arraybackslash}X}

\title{A Method for Learning Value Systems in Generative AI}
\author{
    Andrés Holgado Sánchez\textsuperscript{\rm 1}\\
    Holger Billhardt\textsuperscript{\rm 1}\equalcontrib,
    Sascha Ossowski\textsuperscript{\rm 1}\equalcontrib
}
\affiliations{
    \textsuperscript{\rm 1}CETINIA, University Rey Juan Carlos\\

    Unnumbered Tulipán Street,
    Móstoles, 28933 Madrid, Spain\\
}

\usepackage{bibentry}

\begin{document}

\maketitle

\begin{abstract}
Value‑aware AI systems require explicit computational representations of human values (\emph{groundings}) and their aggregation into value systems in order to align their decisions with ours. As such representations are difficult to elicit, \emph{value learning} seeks to infer them by observing human behaviour.
This work addresses the lack of grounded value learning methods in generative AI: existing approaches typically replicate human preferences without awareness of the multidimensional structure of value alignment, or lack principled value system elicitation methods.
To address these gaps, we adapt a previously validated \emph{value system learning} method to the generative AI setting, which, based on pairwise prompt-response preference data, simultaneously learns: i) an implementation of a grounding for a set of values given by a multi‑objective reward model, and ii) a value system representation in the form of a weighted linear scalarization of the previous grounding model. To ensure that the learned value systems are based on coherent value representations, our algorithm dynamically prioritizes the grounding learning process. 
We evaluate the method against baselines and a contemporary method on prompt-response preference datasets. Results show competitive performance and minimal trade‑offs against the baselines, while improving explainability.
\end{abstract}


\section{Introduction}

Aligning AI systems with human intentions has long been recognized as a central challenge in artificial intelligence~\cite{Anderson2006}. With the widespread deployment of generative AI models --predominantly Large Language Models (LLMs)-- the need for reliable mechanisms to steer their behaviour has become increasingly pressing. The dominant paradigm for this purpose is Reinforcement Learning from Human Feedback (RLHF)~\cite{christiano2023deeprlpreferences}, which aligns a generative model with human preferences by (i) learning a reward model that estimates the \textit{alignment} of its generations with human desires, and (ii) fine-tuning the generative model to maximize this reward using reinforcement learning. 

\begin{figure}
    \centering
    \includegraphics[width=1.0\linewidth]{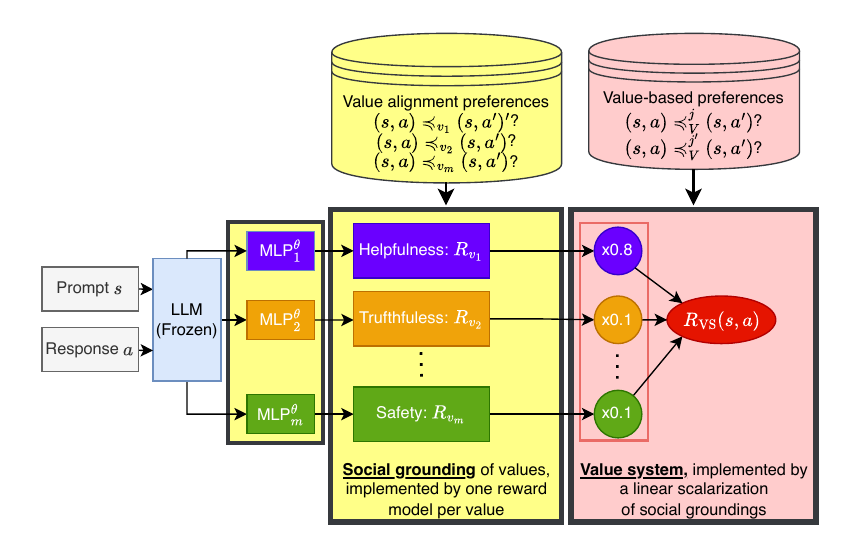}
    \caption{Diagram of the proposed value system learning solution. Based on pairwise preference data reflecting value alignment with multiple human values and the value-based preferences of multiple agents, we learn both a grounding for these values implemented by a multi‑objective reward model and a value system representation in the form of linear scalarization weights.}
    \label{fig:architecture}
\end{figure}
A fundamental limitation of this approach is its inability to recognise the \textit{pluralistic} nature of alignment~\cite{sorensenpluralisticvalues2024valuekleidoscopeanddatasetVALUEPRISM}, i.e. they do not decompose the \textit{value alignment}~\cite{Russell2022alignmentDefinition} of actions in terms of different \emph{human values}. 
The field of \emph{value awareness engineering} argues that AI systems should explicitly represent both values and the value-based preferences of agents (here called \textit{value systems}) in order to be able to reason about them~\cite{valueengineeringAutonomous2023}.
With regard to values, the manual design of value representations is difficult and prone to misspecification~\cite{Sumers2022InstructionsAndDescriptions}, especially if these representations are based on reward models~\cite{inverseRewardHadfieldMisspecification}. As an alternative, authors have proposed \textit{value learning}~\cite{Soares2018ValueLearningProblem} as a means to acquire value representations --hereafter, groundings~\cite{andres2026JAAMASfinal}-- automatically from stated preferences or demonstrations of behaviour. This idea has been explored in some decision-making contexts, where utility or reward functions represent morally relevant concepts and are subsequently used as decision objectives~\cite{leike2020,rodriguez2026reinforcementEthicalEmbedingWeightsRewardRL}.

Following this idea, some works in the generative AI domain (e.g.~\cite{dong-etal-2023-steerlm,wang-etal-2024-arithmetic,Wang2024MAP:Palette}) model alignment with multiple values using independent reward functions, employing a \textit{multi-objective alignment}~\cite{Vamplew2018Human-alignedProblem} framework. 
With regard to modelling the value systems of users, 
in those works it is assumed that users explicitly state their value systems
directly (in the form of value weights or constraints). However, there is evidence that users are often unable to communicate their value system reliably~\cite{Siebert2022liscio}. 
 Other methods approximate a Pareto frontier of aligned models for users to choose from~\cite{zhou-etal-2024-beyondMODPO,robustOnlineDPO}, avoiding directly asking for the users' value systems. However, this approach incurs in additional computational effort and user interactions. Finally, there are approaches that infer value systems from user data. Still, they tend to propose latent~\cite{li-etal-2025-gradient} or opaque value system representations~\cite{wang-etal-2024-interpretableRewardModelingLLM} that have limited interpretability in terms of objectives/values. 


To address these shortcomings, we propose a method for learning an approximation of a complete value system representation suited for preference modelling in generative AI applications, by means of a multi-objective decision making framework. Specifically, we learn a grounding of value labels implemented using reward models, and approximate value systems as weighted linear scalarizations over the previous value representations. Our algorithm jointly learns both representations from prompt--response preference pairs expressed in terms of the alignment of each response with each value, as well as of the alignment with the value system of a user, as illustrated in Figure~\ref{fig:architecture}.
Notably, our optimization dynamically prioritizes the grounding learning process over the value system one, as careful value recognition is a precondition for value system estimation~\cite{Liscio2023BlueskyTrackValueInference}.

This paper is organized as follows. In Section~\ref{sec:relatedwork} we survey related works. In Section~\ref{sec:representingvaluesystems} we propose our value alignment and value system representation approach for generative AI models. In Section~\ref{sec:vslearning} we formalize our \textit{value system learning} problem, and propose a theoretically sound algorithm that approximates a solution to it. Lastly, in Section~\ref{sec:evaluation}, we evaluate our proposal and in Section~\ref{sec:conclusions} we provide conclusions, limitations and avenues for future work.


\section{Related Works}\label{sec:relatedwork}

\subsection{Values in Generative AI} 
Recent work on AI alignment acknowledges that users demand different behavioural properties from generative AI assistants. Two principles arise as the most studied, namely \textit{helpfulness} and \textit{harmlessness}/\textit{safety}~\cite{bai2022constitutionalaiharmlessnessai}. Lately, other principles or goals such as \textit{coherence}~\cite{helpSteerNvidiaWang2024ccby40}, \textit{objectivity}~\cite{ji2024alignanythingtrainingallmodality}, \textit{honesty}~\cite{cui2023ultrafeedback}, or even \textit{humor}~\cite{kopf2023OPENASSISTANTDATASETAlignment} have been explored and aligned for. These have often been considered as alignment targets or subgoals that describe different dimensions of the task of aligning assistant AI models with human values~\cite{dognin2025contextualvaluealignment,Kirk2024TheModelsPRISM}. These alignment targets can be interpreted as domain-dependent ``concepts'' that realize broader human values~\cite{Osman2024}. For instance, the \textit{helpfulness} of a model may be considered a particularization of the more abstract value of \textit{achievement}~\cite{schwartz1992universals}, i.e. the value that fosters the pursuit of personal goals and aims.

\subsection{Alignment With Given Value Systems} Authors typically introduce alignment targets/attributes/criteria (hereafter also referred to simply as \emph{values})\footnote{To simplify the reading, we will refer to all alignment targets or related goals (such as \textit{helpfulness}) as full-fledged ``values'', assuming they have already been identified as relevant alignment targets in the application context considered; see value identification~\cite{Liscio2021}.} with the explicit goal of \textit{steering} the behaviour of LLMs towards different value-based preferences, or \textit{value systems}~\cite{Serramia2018}, specified by the user. In this process, most methods learn or rely on an existing computable representation of the values, i.e. a grounding~\cite{andres2026JAAMASfinal}, that measures the alignment or misalignment of responses with the values considered. This grounding notion is to some extent inherent in recent large language models (LLMs) and can be extracted using careful prompting strategies~\cite{Adams_Hu_Veenhuis_Joy_Ravichandran_Bray_Hoogs_Basharat_2025}. However, for novel application contexts or previously unexplored values, it is generally necessary to train grounding models (normally using \textit{reward models}) from human or AI annotations~\cite{bai2022traininghelpfulharmlessassistantMORLHF}. These value-alignment rewards can then be used to fine-tune a response generator (LLM-based), for instance through RLHF~\cite{christiano2023deeprlpreferences}. 

Notable examples follow the previous paradigm. In~\cite{dong-etal-2023-steerlm}, the authors learn a \emph{multi-attribute reward model} that estimates the degree of alignment of prompt--response pairs with multiple values or attributes. They then fine-tune LLMs through supervised fine-tuning (SFT) to generate responses that adhere to user-specified desired alignment levels over these attributes. Similarly,~\cite{wang-etal-2024-arithmetic} learn a multi-objective reward model via linear regression on quantified value-annotated examples, and subsequently train an LLM to combine these value-specific rewards according to given reward weightings. 
Other approaches aim to control the trade-offs between alignment with different values, such as Controllable DPO~\cite{guo-etal-2024-controllableDPOSFT} and MAP~\cite{Wang2024MAP:Palette}, by enforcing manually specified value-alignment constraints. Finally, some authors focus on obtaining a Pareto front of LLMs aligned for multiple value combinations~\cite{yangRewardsInContextMOAlignment,robustOnlineDPO}. These approaches, however, may expend computational resources on approaching regions of the Pareto front that are ultimately undesirable for the user.

\subsection{Estimating Value Systems} The previous works proposed value-steerable models that learn value representations (groundings), however, to adequately steer their models according to particular users, they expect them to convey their preferences in the form required by their respective architectures (e.g. weights). However, eliciting the value system of a user in terms of context-dependent values, is, in general, a hard problem that has received the name of \textit{value system estimation}~\cite{Siebert2022liscio,Liscio2023BlueskyTrackValueInference}. Classical solutions for this problem include surveys~\cite{schwartz1992universals} and participatory design methodologies~\cite{Ziegfeld_Kox_Akrum_Heijnen_2025valuespecification}. However, it has been shown that value systems elicited in surveys may change depending on the task the user is trying to perform~\cite{Kirk2024TheModelsPRISM}, which makes translating these value system estimations into actionable model behaviour a challenging task.  

Therefore, authors tend to learn value system representations directly from application-oriented data. For example, if a set of models aligned with known value system representations is available, selecting an appropriate model by asking an user to judge responses generated by each model is a natural approach to indirectly estimate her value system~\cite{rameRewardedSoups2023,robustOnlineDPO}, albeit a costly one. A more scalable solution is proposed by Wang et al.~\cite{wang-etal-2024-interpretableRewardModelingLLM}, who learn value-specific reward models and infer prompt-dependent value system representations in the form of scalarization weights via a dedicated deep neural network. While effective, this design lacks interpretability, as the relationship between prompt features and the inferred value system weights is not directly traceable. Relatedly, works such as LoRe~\cite{bose2025lore} and PAL~\cite{chen2025pal} learn latent preference embeddings to efficiently capture personal user preferences (i.e. value systems), but the dimensions of these embeddings do not directly correspond to identifiable values. Unlike these approaches, our method aims to learn a transparent value system representation based on represented values, which allows an agent to reason explicitly on the values and how they combine into a value system.

\section{Value Alignment in Generative AI}

 Following common practice, we model the generative task of any given generative model (e.g. LLM) as a contextual bandit problem~\cite{tennant2025moral} (or MDP, in general), where there is a set of states or instructions (from those of a possibly infinite set $S$, e.g. prompts) for which different actions (from a set $A$, e.g. responses) can be generated or selected. The model is described by a policy $\pi(a\mid s)$ that outputs the most appropriate action $a\in A$ given a state $s\in S$. We define the pair $\tau=(s,a)$ as a \textit{generation} and the set of generations as $\mathcal{T}\subset S\times A$. 

We assume that appropriateness of a generation $\tau=(s,a)$ can be modelled through a quantitative reward function (or \textit{reward model} in the generative AI jargon) $R(s,a)$. Such reward model can be learned from datasets of human annotated examples of preferences~\cite{christiano2023deeprlpreferences} or, directly, from scored outputs to example user inputs~\cite{cui2023ultrafeedback}. The goal of a model $\pi$ that is aware of this reward model is to generate actions that maximize the expected reward.  

In the following, we will be interested in measuring and learning the appropriateness of generations solely in terms of their alignment with human values and, based on that, their appropriateness in terms of the value systems of particular agents, using several reward models. In this paper, we leave out of scope the fine-tuning process of policies with these reward models, that may be carried out using existing methods such as RLHF~\cite{christiano2023deeprlpreferences,bai2022constitutionalaiharmlessnessai}.

\subsection{Representing Value Alignment}

To properly introduce our reward learning objectives, we will adopt the theoretical value alignment representation assumptions previously proposed in~\cite{andresEcai2025} and adapt these to the generative AI context.

\begin{definition}[Value Alignment]\label{def:value-alignment-preference}
     The alignment of a set of generations $\mathcal{T}$ with a value $v_i$ (in general, the alignment preferences with $v_i$) is represented by a weak order $\preccurlyeq_{v_i}$ over $\mathcal{T}$, where $\tau\preccurlyeq_{v_i} \tau'$ means that the generation  $\tau'$ is at least as aligned with value $v_i$ as $\tau$.
\end{definition}
Then, to jointly specify the alignment with a set of values, we formalize the notion of \textit{grounding}.
\begin{definition}[Grounding]\label{def:grounding}
    A \textbf{grounding} of the set of values $V$ is a set of weak orders $\preccurlyeq_{V} = \left\{\preccurlyeq_{v_i}\right\}_{i=1}^m$.  
\end{definition}

To facilitate the computational representation of these value alignment preferences, we propose approximating them using separate utility models for different values. This view has been paramount in works that model human/moral values in general~\cite{manel2022ethical,Serramia2018,montes2022synthesis,Karanik2024}. Thus, we propose a \textit{value alignment function}, $\A_{v_i}$, that quantifies the value alignment of a generation with a value $v_i\in V$, i.e. that represents the relation $\preccurlyeq_{v_i}$: i.e., for all $\tau, \tau' \in \mathcal{T} $:
    $\tau \preccurlyeq_{v_i} \tau' \iff \A_{v_i}(\tau) \leq \A_{v_i}(\tau')$. 

In the generative AI setting, an alignment function can be directly implemented through a reward model: $R_{v_i}(\tau)\equiv \A_{v_i}(\tau)$. Thus, gathering alignment functions for each preference relation, we can represent a grounding $\preccurlyeq_{V}$ with a multi-objective reward model, or \textit{reward vector} $\Rv(s,a)$ such that $\Rv(\tau)=\left(R_{v_1}(\tau),\dots,R_{v_m}(\tau) \right)\equiv\left(\A_{v_1}(\tau), \dots, \A_{v_m}(\tau)\right)$. Given a reward vector, we define the particular grounding that it implements as $\preccurlyeq_\Rv=\{\preccurlyeq_{R_{v_i}}\}_{i=1}^m$, i.e. $\tau\preccurlyeq_{R_{v_i}}\tau'\iff R_{v_i}(\tau)\leq R_{v_i}(\tau')$. In occasions, we will refer to this reward vector also as the \textit{grounding function/model} to emphasize its role as the implementation of a grounding. 

Importantly, our reward functions do not aim to capture the whole meaning of the given values, rather, we are estimating the \textit{alignment} of actions with these value based on the proposed multi-objective model.

\subsection{Representing Value Systems}\label{sec:representingvaluesystems}
A \textit{value system} expresses the general importance assigned to each value by a certain person or group (in general, an ``agent'') in a specific context. Given a grounding, it represents the individual (value-based) preferences over generations of an agent.

\begin{definition}[Value system]\label{def:value-system}
The \textbf{value system} of an agent $j$, based on a grounding $\preccurlyeq_{V}$, is determined by a weak order $\preccurlyeq^j_{V}$ over $\mathcal{T}$. If $\tau \preccurlyeq^j_{V} \tau'$, we say that $\tau$ is equally or more aligned than $\tau$ with $j$'s value system.
\end{definition}


The idea is that the value system preferences of an agent are based on the individual value preferences. Following other works in the field~\cite{Serramia2018,rodriguez2026reinforcementEthicalEmbedingWeightsRewardRL}, we assume that the value system of an agent can be approximated by a linear combination of the alignment of a generation with the values. Thus, our value system model is \textit{welfarist utilitarian}~\cite{Sinnott-Armstrong2019-SINC-5}, i.e. values are conceived as different sources of good, and each agent's value system weighs their relative importance.

\begin{definition}[Value System Function]\label{def:value-system-alignment-function}
Let $j$ be an agent with a value system $\preccurlyeq^j_{V}$ and a grounding represented by a reward vector $\Rv$. The function 
$\A_{W_j,{\Rv}}(\tau) = W_j\cdot(R_{v_1}(\tau), \dots, R_{v_m}(\tau))^T$ is a \textbf{value system function} for $j$ if it represents $\preccurlyeq^j_{V}$ over $\mathcal{T}$, i.e., for all $\tau,\tau' \in \mathcal{T}$:
$$\A_{W_j,{\Rv}}(\tau) \leq \A_{W_j,{\Rv}}(\tau') \iff \tau \preccurlyeq^j_{V} \tau'$$ where $W_j=(w_j^{v_1},\dots,w_j^{v_m})$ are called \textit{value system weights}. These are bounded in the unit $(m-1)$-simplex\footnote{As we treat values as sources of good, we do not consider value systems that willingly demote values by using negative weights. If demoting a value is expected to occur, the designer should instead include another value in the model as the negated version of the original. An example of such case is the technical value of response \textit{complexity}~\cite{helpSteerNvidiaWang2024ccby40}, which in some situations might be undesirable (when a response is for a audience of children, for instance). Instead of using a negative value system weight, we suggest adding the contrary value of \textit{simplicity}, to represent the abstract goal of reducing complexity.}: $W_j\in \Delta^{m-1}$, i.e. $W_j\in {(0,1)}^m$, $\sum_{i=1}^m w_j^{v_i}=1$.
\end{definition} 

We can interpret the previous definition as having a \textit{value system reward} for each agent that is a linear scalarization~\cite{linearscalarizedMOMDP2013} of the multi-objective reward vector $\Rv$ with weights $W_j$:
$$R_j(\tau)=W_j\cdot \Rv(\tau)^T$$ 

Finally, we denote the value system implemented through the value system alignment function $\A_{W_j,{\Rv}}$ (with weights $W_j$ and reward vector $\Rv$) by $\preccurlyeq^{W_j}_{\Rv}$, i.e. $\tau \preccurlyeq^{W_j}_{\Rv} \tau'\iff W_j\cdot \Rv(\tau) \leq W_j\cdot \Rv(\tau') $.

We emphasize that our value system representation is a mere approximation of the true value systems of agents, which might be much more complex and context-dependent. We chose a linear representation that eases the interpretability of our model.

\section{Learning Value Systems}\label{sec:vslearning}

In this section, we present our proposal: a model and algorithm for learning simultaneously a reward vector $\Rv$ and a weight vector $W$, which approximate a grounding of the values considered and the value system of a certain agent, respectively. Our learning approach relies on a dataset of annotated preferences for both, values and the value system of an agent.

\subsection{Data Assumptions and Scope }

Within the field of generative AI, even though datasets that identify individual annotators or agents exist~\cite{kopf2023OPENASSISTANTDATASETAlignment,Kirk2024TheModelsPRISM}, the majority of datasets that label examples along multiple alignment dimensions or values agglomerate the opinions of distinct individuals with potentially divergent preferences~\cite{helpSteerNvidiaWang2024ccby40,ji2024safetydataset,ji2024alignanythingtrainingallmodality}. Our work seeks to be applicable in the latter cases as well, assuming that the dataset preferences are those of an hypothetical abstract agent representing the collective. 

This assumption may raise concerns about the usefulness of the proposed approach, as it might attempt to reconcile potentially divergent opinions into a single representation. However, with regard to value recognition in particular, we argue this view might not be problematic. Namely, although we consider agents may hold divergent interpretations of the values under consideration (for example, the meaning of “freedom”), within delimited domains we can reasonably assume that most agents understand values in broadly similar manner. That is, there exists a degree of consensus regarding the meaning of values, which we refer to as a social grounding~\cite{andresEcai2025}. 

In the generative AI domain, the existence of such social groundings is often implicitly assumed. We identify two main reasons for this. First, most alignment-related datasets justify the inclusion of multiple alignment criteria as universal quality dimensions and provide explicit definitions for each, which condition the preference annotation process~\cite{bai2022constitutionalaiharmlessnessai}. Second, some works explicitly average different opinions over the same utterances to improve ``robustness'' of the annotations~\cite{helpSteerNvidiaWang2024ccby40}, to then train a shared ``grounding-function-like'' model.

With regard to value systems, it is clear that even in the same application domain, agents may have diverse value-based preferences (value systems). 
Our work, however, performs an implicit aggregation~\cite{leraleri2024aggregation} of potentially divergent value systems. Doing this has been shown to have limitations in terms of the final representativeness of value-based preferences, specially in datasets with heterogeneous agents~\cite{andresEcai2025,pmlr-v235-chakraborty24b}. We acknowledge potential limitations of the representativeness that we achieve in these cases, but argue that our approach is still valuable, as the learned value system can be interpreted as a transparent representation of the overall, empirical, importance given to each value according to the dataset annotators.

In conclusion, in this paper, we consider that the value alignment preferences observed in our dataset correspond to a social grounding in a delimited context (e.g. text generation) and the value system preferences pertain to  some abstract agent. When data annotations originate from a group of agents with potentially heterogeneous value systems, then our learning approach would try to identify the aggregated value system that best represents the collective preferences of that group. 

Importantly, due to the limitations mentioned above, if our model is learned from heterogeneous datasets, the resulting value system reward (given by $W\cdot \Rv$) should not be intended for use as a stand-alone objective in RLHF fine-tuning. In such cases, it would likely bias the generative model toward majority opinions. Instead, we recommend using the learned grounding function $\Rv$ for multi-objective fine-tuning. Specifically, generative models should be trained in a way that allows their immediate alignment with respect diverse value systems, following approaches similar to prior work~\cite{robustOnlineDPO,helpSteerNvidiaWang2024ccby40,Jang2023PersonalizedMerging}. Our approach offers an additional advantage: it enables the exploration of value systems in the vicinity of the learned one, that acts as the representation of the preferences of the majority.

\subsection{Problem Definition}
Considering the above, our \textit{value system learning} problem consists of   \textit{learning} a grounding implementation together with a representation of the value system of a certain (possibly abstract) agent $j$ that best approximates the social grounding $\preccurlyeq_V$ and value system $\preccurlyeq_V^j$ enacted in a observed preference dataset. In particular, we assume access to examples of value alignment and value system preferences over pairs of generations. Formally, we assume a dataset, $\DS$, composed by entries of the type $(\tau, \tau',y^j_V,y_{v_1},\dots,y_{v_m})$, where each \textit{label} $y_{\_}\in \{0,0.5,1\}$ indicates the preference of $\tau$ over $\tau'$ (1), the contrary case (0) or indifference (0.5). Specifically, label $y^j_V$ indicates the preference according to the agent's value system, and labels $y_{v_i}$, $i=1,\dots,m$ represent the preference regarding the alignment of each generation pair with respect to each value $v_i$. As mentioned above, we assume that the latter labels correspond to the social grounding of the values.

Given the analysis in Section~\ref{sec:representingvaluesystems},  we propose reducing our learning problem into two tasks i) learn a quantitative reward vector $\Rv$ such that the grounding $\preccurlyeq_\Rv$ implemented by $\Rv$ approximates the real social grounding $\preccurlyeq_V$ and ii) learn a vector of value system weights $W$ such that the value system $\preccurlyeq_\Rv^{W}$ approximates the real value system $\preccurlyeq_V^j$. 

We now provide ways to quantitatively \textit{measure} the appropriateness of potential candidate solutions. For that, we quantify the alignment preference differences between our models and the given data. To measure preference differences in general, we propose the \textit{discordance} between preference relations ($\preccurlyeq^1$, $\preccurlyeq^2$) as the proportion of ordered pairs of generations in a set $S\subseteq \T\times \T$ that are ranked differently:

\begin{equation}\label{eq:discordance}
    d_{S}\left(\preccurlyeq^1, \preccurlyeq^2\right) = \frac{1}{\abs{S}}\sum_{\substack{(\tau,\tau')\in S}}\mathds{1}\left(\left({\tau\preccurlyeq^1 \tau'}\right) \not\equiv \left(\tau \preccurlyeq^2 \tau'\right)\right)
\end{equation}  

In our settings, the aim is to learn grounding functions and value system weights with minimal discordance to the original data (represented in the dataset $\DS$). Thus,  
we can define our value system learning problem as a nested optimization problem:

\begin{align}\label{eq:bilevel}
(W^*,\Rv^*) &\in \argmin_{W,\Rv'=(R_{v_1},\dots,R_{v_m})}\ {d_{\DS}\left({
\preccurlyeq^W_{\Rv'},
\preccurlyeq^j_V
}\right)}
\nonumber
\\
\text{subject to} &\quad R_{v_i}' \in \argmin_{R_{v_i}} \ 
{d_{\DS}\left({
\preccurlyeq_{R_{v_i}}, \preccurlyeq_{v_i}
}\right),}\nonumber\\
&\quad \forall i\in\{1,\dots,m\}
\end{align}
The formulation prioritizes the minimization of the average discordance on the learned grounding functions, before improving the value system estimation. This ensures that the learned value system weights are based on a coherent social grounding approximation.

It should be noted that in Problem~\ref{eq:bilevel_loss}, the discordance can be directly calculated by comparing the preference predictions of the learned models with the original preference relations on the generation pairs included in the dataset $\DS$. To do so, we discretize the score differences predicted by our model for each pair of generations to predict strict preference ($\tau\prec \tau'/\tau\succ \tau'$) or indifference ($\tau \simeq \tau'$), which is treated as predicting both $\tau'\preccurlyeq\tau$ and $\tau'\succcurlyeq\tau$. Then, we check the agreement of the discretized predictions with the corresponding dataset label $y\in\{0,1,0.5\}$. Details on how the predictions of the model are discretized and the data-based discordance is calculated are given in supplementary material, Section A.3.

\subsection{Algorithm}

Our algorithm is an approximate solution to Problem~\ref{eq:bilevel} that versions the classical deep reward learning method based on the Bradley-Terry (BT) model~\cite{bradleyTerryModel1952}. The solution is represented by several neural networks. First, a reward vector network $\Rv^\theta$ with parameters $\theta$ implements an approximation the social grounding in the data. Second, a linear layer given by certain value system weights $W^\omega$ that are parametrized with $\omega \in \R^{m}$ through a \textit{softmax} calculation: $W^\omega=(w^{v_1}, \dots, w^{v_m}) = \frac{\exp \omega}{\sum \exp(\omega)} $, so $W^\omega$ remains in the simplex.

For learning purposes, we approximate the value alignment and value system preferences in a differentiable manner with the BT model. Namely, we use the sigmoid ($\sigma$) of the reward difference between two generations that share the same initial prompt: $p(\tau \succ \tau' \mid R) = \sigma(R(\tau)-R(\tau'))$, as an approximation of the degree of preference (given as a probability between $0$ and $1$) of preferring the generation $\tau$ over $\tau'$. To learn with this model, we modify the parameters of $R$ so that the values of $p(\tau \succ \tau' \mid R)$ are closer to the preference labels $y\in \{0,0.5,1\}$ by minimizing a binary cross-entropy loss~\cite{christiano2023deeprlpreferences}. Minimizing this loss, in Eq.~\ref{eq:crossentropy}, discordance between the model $R$ and the dataset will also be decreased.

We add a modification to the loss using the centering mechanism proposed in~\cite{eisenstein2024helpingherdingrewardmodel}, which reduces the magnitude of the estimated rewards and helps ensure that reward scales are comparable across values. This mechanism introduces a hyperparameter $r$, which we set to the recommended value of $r = 0.01$.

\begin{align}\label{eq:crossentropy}
    \mathcal{L}(\tau,\tau',y \mid R) = &-y\log p(\tau \succ \tau' \mid R) \nonumber\\&- (1-y)\log\left(p(\tau \prec \tau' \mid R)\right)\nonumber\\&+r\left( R(\tau)+R(\tau')\right)^2
\end{align}
 
 Then, we propose a loss associated to the discordance of our model with the alignment preference examples for each value $v_i\in V$  in our dataset (Eq.~\ref{eq:losschr}).
 \begin{equation}\label{eq:losschr}
    \mathcal{L}_{v_i}^{\theta}(\DS)=
    \sum_{(\tau,\tau',y_{v_i}) \in \DS} \frac{\mathcal{L}( \tau, \tau', y_{v_i}\mid \Rv^\theta)}{{\abs{\DS}}}
\end{equation}

We consider all the previous losses together to define a \textit{grounding (discordance) loss} given by: $\mathcal{L}_{V}^{\theta}(\DS)=\left(\mathcal{L}_{v_1}^{\theta}(\DS),\dots,\mathcal{L}_{v_m}^{\theta}(\DS)\right)$.

And we define a loss associated to value system discordance (Eq.~\ref{eq:lossrepr}).

\begin{equation}\label{eq:lossrepr}
\resizebox{0.9\linewidth}{!}{
    $\displaystyle
    \mathcal{L}_{\VS}^{\omega,\theta}(\DS)= \sum_{(\tau,\tau',y^j_V) \in \DS} \frac{\mathcal{L}( \tau, \tau', y^j_V \mid W^\omega\cdot\Rv^\theta )}{\abs{\DS}}
    $}
\end{equation}

A naïve approach to use these losses to solve Problem~\ref{eq:bilevel} would be first minimizing the average of the grounding discordance losses $\mathcal{L}_{V}^{\theta}$, and then, with $\theta$ frozen, learn the weights that minimize the value system discordance loss $\mathcal{L}_{\VS}^{\omega,\theta}$. However, it has been shown that multiple grounding functions are compatible with a low discordance while only a subset of those can also be effectively combined linearly to approximate value system~\cite{andresEcai2025} . 

Therefore, we propose a \textit{simultaneous} learning approach that balances both goals. The idea is treating Problem~\ref{eq:bilevel} as a constrained optimization problem, where the constraint is minimizing the inner grounding discordance objective. In this formulation (Problem~\ref{eq:bilevel_loss}), we introduce the constraints that the grounding discordance losses in Eq.~\ref{eq:losschr} should reach an (unknown, approximated) per-value \textit{loss target} $\mathcal{L}^*_{v_i}$ that depends on the capabilities of the available model. These targets $\{\mathcal{L}^*_{v_i}\}_{i=1}^m$ are estimated during training by an ``exponentially weighted \textit{minimum}'' of observed losses during training (see step 3 of the \textbf{Algorithm description}).

\begin{align}\label{eq:bilevel_loss}
(\omega^*,\theta^*) \in \argmin_{W,\Rv'} 
&\ \mathcal{L}^{\omega,\theta'}_{\VS}\left(\DS\right) \nonumber\\
\text{subject to} 
\quad \theta' 
&\in \{\theta \mid \mathcal{L}^{\theta}_{v_i}\left(\DS\right)
 \leq \mathcal{L}^*_{v_i}\} \nonumber\\
&\hspace{-0.9em}\forall i\in\{1,\dots,m\}
\end{align}

To approach Problem~\ref{eq:bilevel_loss}, we obtain its Lagrangian  (Eq.~\ref{eq:lagrangian}). It introduces Lagrange multipliers $\pmb{\lambda}=(\lambda_1, \dots,\lambda_m)\in \R^+$.  Then, the optimization consists of finding a saddle point for the \textit{min-max} goal in Eq.~\ref{eq:MinMaxLag}. 

\begin{align}\label{eq:lagrangian}
    \mathcal{L}^{\omega,\theta}_{\pmb{\lambda}}(\DS) &= \mathcal{L}_{\VS}^{\omega,\theta}(\DS)\nonumber\\ &+ \sum_{i=1}^m\lambda_i  \left(\mathcal{L}^{\theta}_{v_i}(\DS)-\mathcal{L}^*_{v_i}\right)
\end{align}
\begin{align}\label{eq:MinMaxLag}
    \min_{\theta,\omega}\left[ \max_{\pmb{\lambda}>0}\left[\mathcal{L}^{\omega,\theta}_{\pmb{\lambda}}(\DS)\right]\right]
\end{align}

Assuming that the reward vector $\Rv^\theta$ is an affine transformation (e.g. obtained by training only an additional single-layer perceptron over embeddings extracted by a ``base'' LLM), Problem~\ref{eq:bilevel_loss} is convex, and with additional minor assumptions, strong duality holds~\cite{boyd2004convex} , which justifies the theoretical use of dual ascent~\cite{boyd2011ADMMreferenceOfDualAscent} to find an optimal solution to Problem~\ref{eq:MinMaxLag} (for fixed loss targets). Instead of dual ascent, which requires the exact minimization of the Lagrangian at every training iteration with respect to $\omega$ and $\theta$, we approximate a solution with a modified gradient descent-ascent (GDA) algorithm~\cite{zamani2022convergencerateanalysisgradient}, which converges to a local Nash equilibrium and even the global optimum with certain learning rate choices~\cite{fallahSGDAGlobalconvergencewithSMOOTHSTROGCONVEX}. See supplementary material, Section A.2 for proofs of the previous theoretical aspects.

Before introducing the algorithm, we add some technical modifications. First, to guarantee that the learning step sizes on the Lagrangian remain constant, we add an auxiliar ``value system multiplier'', $\lambda_\VS$, which is calculated so that the sum of the multipliers including this one have an average value of $1$\footnote{A similar multiplier normalization process was proposed for Safe RLHF methods~\cite{Dai2024}.}. Second, to avoid any multiplier (including $\lambda_\VS$) from vanishing, we force them to be over a certain user-defined threshold $\lambda_{0}>0$. To satisfy the previous constraints, we use auxiliary parameters $\bar{\pmb\lambda}=(\bar{\lambda}_1,\dots,\bar{\lambda}_m,\bar{\lambda}_\VS)\in\R^{m+1}$ to parametrize the ``grounding multipliers'' $\pmb\lambda=(\lambda_1,\dots,\lambda_m)$ and the value system one $\lambda_\VS$, using: $(\lambda_1,\dots,\lambda_m,\lambda_\VS) =\frac{(m+1)\exp\bar{\pmb\lambda}}{\sum_{\bar{\lambda}\in\bar{\pmb{\lambda}}}\exp{\bar{\lambda}}}(1-\lambda_0) +\lambda_0$. Finally, for theoretical reasons, we add $\ell_2$ regularization to the Lagrangian, affecting all parameters $\omega,\theta,\bar{\pmb\lambda}$. The Lagrangian used in our algorithm with the previous modifications is in Eq.~\ref{eq:lagrangianV2}.

\begin{align}\label{eq:lagrangianV2}
    \bar{\loss}^{\omega,\theta}_{\bar{\pmb{\lambda}}}(\DS) &= \lambda_\VS\mathcal{L}_{\VS}^{\omega,\theta} (\DS) + \gamma_\omega\|\omega\|^2+ \gamma_\lambda\|\bar{\pmb\lambda}\|^2\nonumber\\ &+ \sum_{i=1}^m\lambda_i  \left(\mathcal{L}^{\theta}_{v_i}(\DS)-\mathcal{L}^*_{v_i}\right) + \gamma_\theta\|\theta\|^2
\end{align}


\paragraph{Algorithm description.} After random parameter initialization, our algorithm proceeds by repeatedly executing the following three steps over all batches in the dataset, for $N$ epochs\footnote{An epoch here is a pass over the complete dataset.}:
\begin{enumerate}
    \item \textbf{Gradient descent.} Calculate gradients of Eq.~\ref{eq:lagrangianV2} 
    with respect to the model parameters ($\omega$ and $\theta$) for several batches (accumulating them). Then, perform gradient \textit{descent} step to update $\omega$ and $\theta$ .
    
    \item \textbf{Gradient ascent.} We perform a gradient \textit{ascent} step on the Lagrange multiplier parameters $\bar{\pmb{\lambda}}$ based on the Lagrangian. In practice, for this update, we compute the gradient of a modified version of the Lagrangian in Eq.~\ref{eq:lagrangian} with grounding discordance losses measured on previous batches (\text{BS}): $$\nabla_{\bar{\pmb\lambda}}\left[\sum_{i=1}^m\lambda_i \max(\loss^\theta_{v_i}(\text{BS})-\loss^*_{v_i},0) +\gamma_\lambda\|\bar{\lambda}\|^2\right]$$ 
    With the ``$\max$'' projection, this gradient is approximately zero (up to the effect of $\ell_2$ regularization) for those multipliers whose associated grounding discordance losses, computed over recent batches, have already met their target values. This design ensures that the primary objective remains reducing grounding discordance losses as much as possible, even beyond the target estimations. The multipliers can still may still decrease slightly due to the $\ell_2$ regularization term, which helps correct potential overestimations. It is noticeable that we do not update the value system multiplier $\lambda_\VS$ based on the value system discordance loss. However, when the gradients of the grounding multipliers ($\pmb\lambda=(\lambda_1, \dots,\lambda_m$) vanish, the $\ell_2$ regularization progressively reduces the gap between $\lambda_\VS$ and $\pmb{\lambda}$. Thus, only once the grounding discordance losses have reached their targets, the relative priority of minimizing the value system discordance loss is subject to a slight increase.

    \item \textbf{Loss target update.} After repeating steps 1-2 for a number of steps and batches (let this number be $u>0$), we re-estimate the grounding discordance loss targets that are attainable with our reward model $\Rv^\theta$. To do so, we employ a ``exponentially weighted \textit{minimum}'' with coefficient $\eta\in (0,1)$ (close to $1$) of the averages of the losses obtained during each period of $u$ steps. Namely, given the calculation of the loss associated with each value $v_i$ across $u$ different batches $\{\text{BS}_k\}_{k=1}^u$, and its current loss target estimation $\mathcal{L}^*_{v_i}$, the update is:
    $\mathcal{L}^*_{v_i} \gets \eta\mathcal{L}^*_{v_i}+(1-\eta)\min\{\mathcal{L}^*_{v_i},\frac{1}{u}\sum_{k=1}^u\mathcal{L}_{v_i}^\theta (\text{BS}_k)\}$. We set the initial value of the target losses as the average of the losses obtained in the first $u$ batches.
\end{enumerate}

We supply a fully-detailed algorithm pseudocode in the supplementary material.

\section{Evaluation}\label{sec:evaluation}

We evaluate our proposal by comparing its performance against two baselines and the Armo-RM method~\cite{wang-etal-2024-interpretableRewardModelingLLM}, a competitive reward model with the most similar architecture to our approach we found in the literature. We analyse 2 datasets: \Ultra~\cite{cui2023ultrafeedback} and the text-to-text task in \PKU~\cite{ji2024alignanythingtrainingallmodality}. We chose these datasets since they have been widely used for LLM alignment and because their structure is close to our dataset assumptions: indeed, in \PKU the annotators are not identified and in \Ultra, the annotator is a single GPT-4 agent.

\subsubsection{Methods.} Here we detail each method characteristics and training setting:

\begin{itemize}

\item \textbf{Armo-RM.} A pretrained multi-objective reward model that predicts alignment with 19 \textit{attributes} (that we treat as values) from six datasets. It uses a frozen \textit{Llama3-8B-v0.1} backbone and a linear head over the last hidden state acting similar to our reward vector $\Rv$ (grounding function). A scalar reward is obtained by linearly combining these outputs with weights $W_{\text{arm}}(p)\in[-1,1]^{19}$, produced by a prompt-conditioned three-layer MLP trained with MSE.

\item \textbf{VSL-RM (ours).} We use the same reward vector architecture as Armo-RM (a linear layer over the frozen \textit{Llama3-8B-v0.1} hidden state). We train it with our algorithm from scratch on both datasets.

\item \textbf{BT-RM.} A single-objective model that minimizes only the value system discordance loss (Eq.~\ref{eq:lossrepr}), without observing grounding discordance losses. It shares the same architecture and parameters as VSL-RM, and we also train one version per dataset. It serves as a lower bound on attainable value system discordance with the used architecture.

\item \textbf{SEQ-RM. } Shares the same architecture as VSL-RM but is trained sequentially: first to solely minimize grounding discordance (updating $\theta$), then minimizing the value system discordance (updating $\omega$ with $\theta$ frozen) for an additional 50\% of epochs and higher learning rate. It serves both as a control baseline for comparing sequential vs. simultaneous training and as oracle for minimum attainable grounding discordance with the used architecture.
\end{itemize}

\subsubsection{Training details.} The datasets are split into training, validation, and test sets. The validation data is used for hyperparameter selection, while test data is used only for final evaluation. All models are trained with four random seeds (shared across methods), except Armo-RM, which is used as provided on HuggingFace\footnote{Armo RM model: \url{https://huggingface.co/RLHFlow/ArmoRM-Llama3-8B-v0.1}, version from May 11th, 2026.}.

We trained the models (except Armo-RM) using approximately the same set of common hyperparameters selected with a standard search method. The baseline models rely primarily on standard hyperparameters (e.g., batch size), which have limited impact on performance in this context, so we keep them aligned with those used for VSL-RM without further tuning. The only exception is SEQ-RM, for which we increase the value system learning rate to compensate for the smaller number of value system training iterations. 

All experiments were repeated with four random seeds. We observed little variation across runs, and therefore decided not to increase the number of training seeds. As optimizer, we used a paged, 32 bit ``AdamW'' one. Further details on hyperparameters, dataset preprocessing, and hardware are provided in supplementary material, Section A.4.

\paragraph{Metrics.} 
For an easier interpretation of results, we report grounding and value system \textit{accuracies}, which we define as the complement of the grounding and value system discordances, respectively:
\[
\text{acc}_{\DS}(R_{v_i}, \preccurlyeq_{v_i}) = 1 - d_{\DS}(\preccurlyeq_{R_{v_i}}, \preccurlyeq_{v_i}),
\]
\[
\text{acc}_{\DS}((W,\Rv), \preccurlyeq_V^j) = 1 - d_{\DS}(\preccurlyeq^W_{\Rv}, \preccurlyeq_V^j).
\]
Accuracy tends to $1$ when each model perfectly reproduces the preferences of the data (corresponding to a discordance of $0$), and tends to $0$ when discordance raises to its maximum ($1$).

\subsection{Results}
Figure~\ref{fig:curves} presents training curves for the \emph{average grounding accuracy} (denoted by $\textbf{AGA} = \sum_{i=1}^m \frac{1}{m}\text{acc}_{\DS}(R{v_i}, \preccurlyeq_{v_i})$) and the \emph{value system accuracy} (denoted by $\textbf{VSA} = \text{acc}_{\DS}((W,\Rv), \preccurlyeq_V^j)$). The curves depict the evolution of these accuracies throughout training, measured on the validation splits of each dataset. We do not present the curves for Armo-RM, as it was originally pretrained.

As expected, BT-RM achieves highest VSA, as it is optimized solely for this objective, while VSL-RM clearly outperforms it in AGA due to the inclusion of the grounding discordance loss. Interestingly, AGA increases rapidly in both methods. Likely, this occurs because of our architecture where the value system reward is a positive combination of a  reward vector (corresponding to the different values). Thus, a better prediction of the preferences of the value system correlates with a better value alignment estimation. Also in the case of BT-RM a reward vector is learned. However, this vector would not directly correspond to the underlying values. Still, the learned vector is indirectly correlated with the original values. This explains the fact that AGA increases rapidly for BT-RM (due to the positive correlations) and then drops (since the reward vector does not actually represent the original values). In VSL-RM, AGA still improves over the learning epochs.

\begin{figure}[t]
    \centering
    \includegraphics[width=0.95\columnwidth]{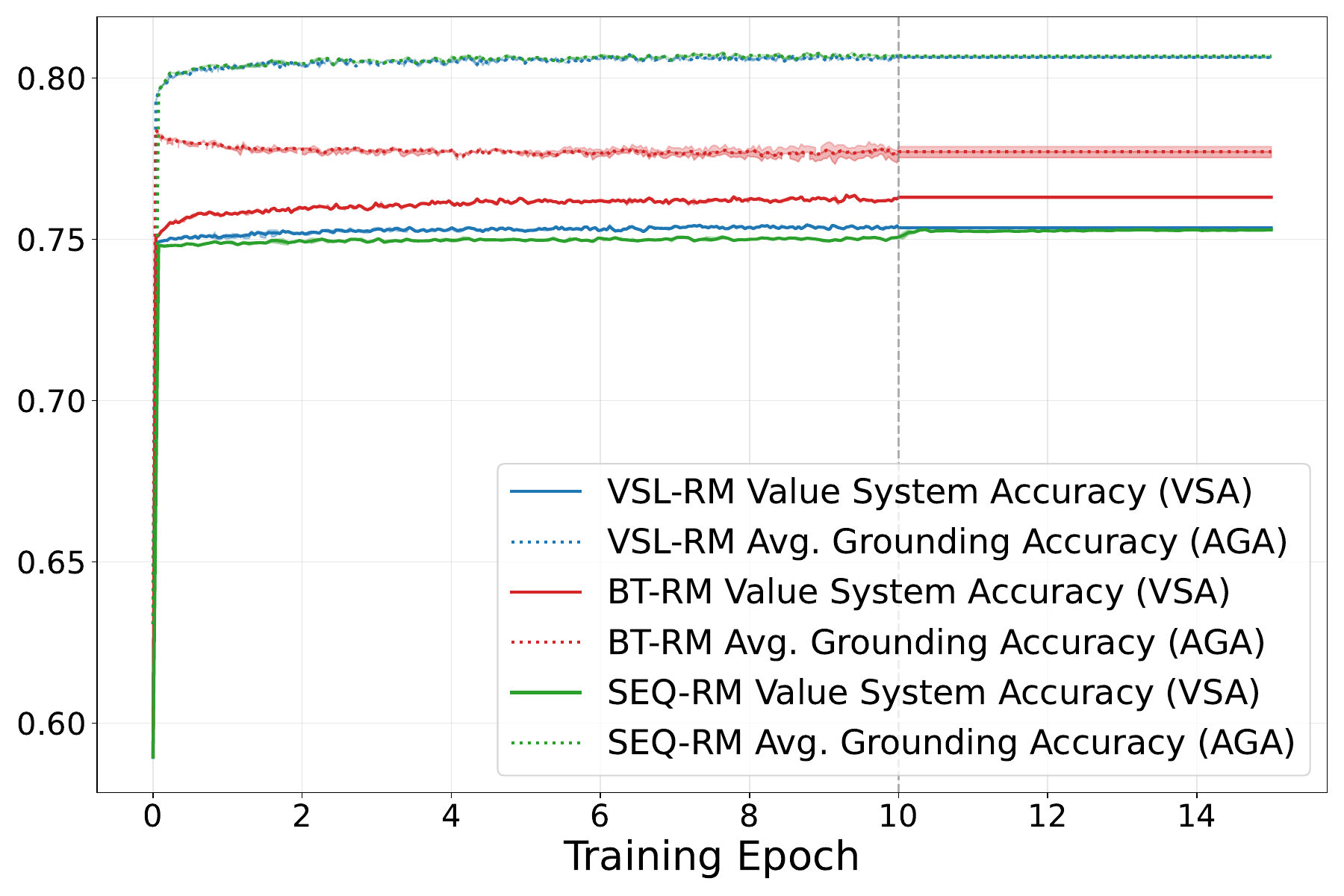}
    \includegraphics[width=0.95\columnwidth]{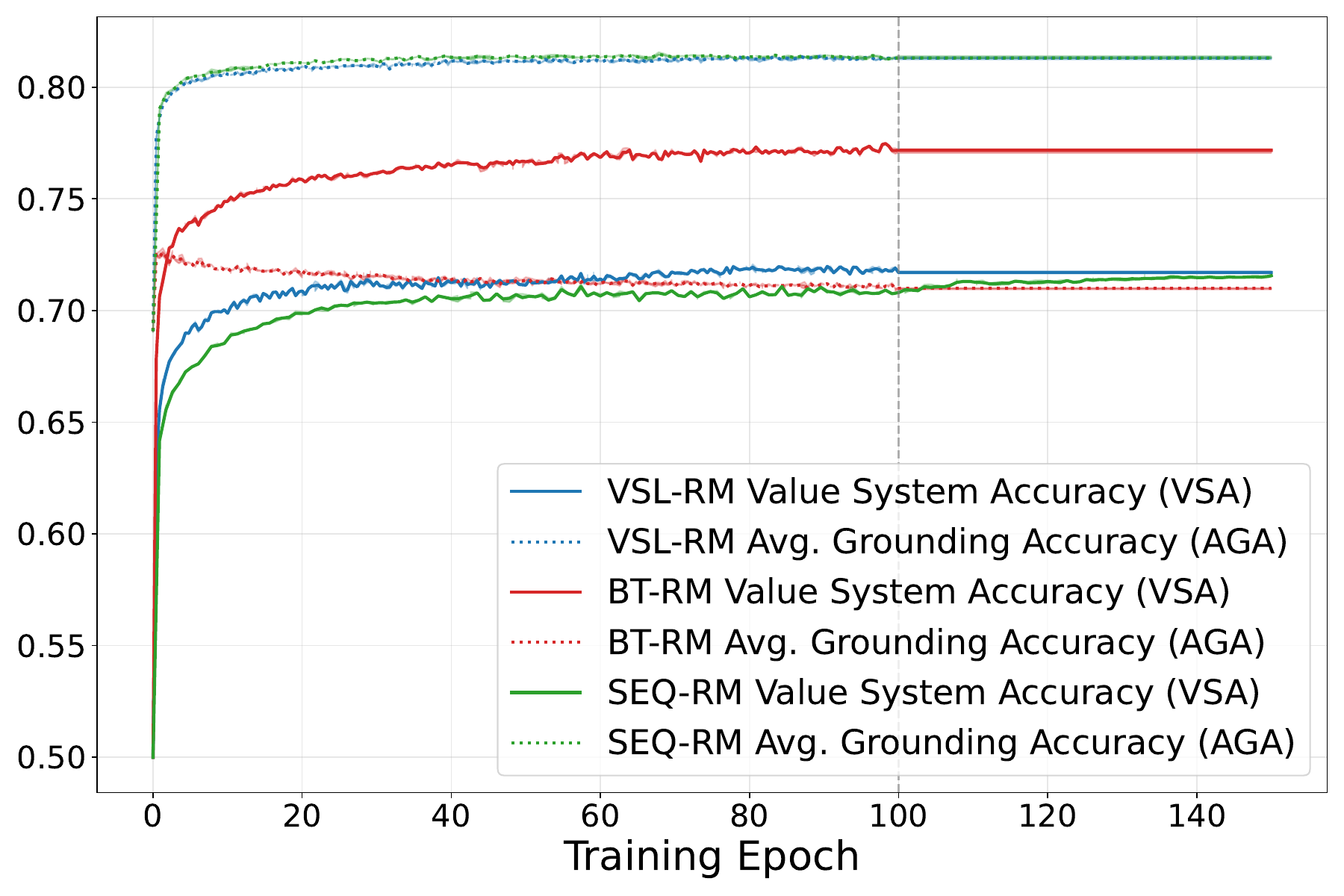}
    \caption{Average grounding and value system accuracy per training method (\textcolor{blue}{VSL-RM}, \textcolor{red}{BT-RM}, \textcolor{teal}{SEQ-RM}) in the validation splits of \Ultra (Top) and \PKU (Bottom). It might be difficult to appreciate, but AGA by VSL-RM and SEQ-RM almost coincide. The vertical line indicates the training moment where VSL-RM an BT-RM end their training and SEQ-RM starts its second phase to maximize VSA.}
    \label{fig:curves}
\end{figure}

\begin{table*}[t]
\centering

\begin{tabular}{lrrrr|rr}
\toprule
Method & \multicolumn{1}{r}{Helpfulness} & \multicolumn{1}{r}{Honesty} & \multicolumn{1}{r}{Truthfulness} & \multicolumn{1}{r}{Instruction Following} & \multicolumn{1}{|r}{AGA} & \multicolumn{1}{r}{VSA} \\
\midrule
Armo-RM & \makecell[r]{\textbf{0.826} } & \makecell[r]{\textbf{0.814} } & \makecell[r]{\textbf{0.840} } & \makecell[r]{\textbf{0.824} } & \makecell[r]{\textbf{0.826} } & \makecell[r]{0.754 } \\
\midrule
BT-RM & --- & --- & --- &--- & ---& \makecell[r]{\textbf{0.768} \\ $\pm$0.0002} \\
VSL-RM & \makecell[r]{\underline{0.812} \\ $\pm$0.0005} & \makecell[r]{\underline{0.793} \\ $\pm$0.0002} & \makecell[r]{0.786 \\ $\pm$0.0007} & \makecell[r]{\underline{0.815} \\ $\pm$0.0003} & \makecell[r]{0.802\\ $\pm$0.0002} & \makecell[r]{\underline{0.760} \\ $\pm$0.0003} \\
SEQ-RM & \makecell[r]{\underline{0.813} \\ $\pm$0.0003} & \makecell[r]{\underline{0.793} \\ $\pm$0.0003} & \makecell[r]{{\underline{0.788} }\\ $\pm$0.0008} & \makecell[r]{\underline{0.815} \\ $\pm$0.0003} & \makecell[r]{\underline{0.802} \\ $\pm$0.0002} & \makecell[r]{0.758 \\ $\pm$0.0002} \\
\bottomrule
\end{tabular}%

\caption{Average grounding accuracy and value system accuracy (\Ultra).}
\label{tab:group-summary}
\end{table*}
For SEQ-RM, both AGA and VSA improve during the first phase in both datasets, with VSA computed using uniform value system weights ($1/m$)\footnote{VSA increases here due to the architectural effect described before: a good estimation of value alignemnt preferences, combined with positive weights, leads to a better estimation of value-based preferences.}. Despite its simultaneous learning approach, VSL-RM matches the AGA achieved by SEQ-RM, which is optimized solely for grounding discordance. This indicates that VSL-RM effectively prioritizes learning grounding functions with the lowest discordance, as intended in Problem~\ref{eq:bilevel}. In the second phase of SEQ-RM, VSA increases further, almost matching the one achieved with VSL-RM. We will show the gap is still statistically significative, by evaluating the test datasets.

\begin{table*}[t]
\centering
\begin{tabular}{lrrrrr|rr}
\toprule
Method & \multicolumn{1}{r}{Prompt Following} & \multicolumn{1}{r}{Objectivity} & \multicolumn{1}{r}{Clarity} & \multicolumn{1}{r}{Information Richness} & \multicolumn{1}{r}{Safety} & \multicolumn{1}{|r}{AGA} & \multicolumn{1}{r}{VSA}  \\
\midrule
Armo-RM & \makecell[r]{0.715} & \makecell[r]{\underline{0.798}} & \makecell[r]{0.609} & \makecell[r]{0.647} & \makecell[r]{0.837} & \makecell[r]{0.721} & \makecell[r]{0.685} \\
\midrule
BT-RM & --- & --- & --- &--- & ---& ---&\makecell[r]{\textbf{0.773}\\ $\pm$0.0008} \\
VSL-RM & \makecell[r]{\textbf{0.765} \\ $\pm$0.0029} & \makecell[r]{\underline{0.798} \\ $\pm$0.0030} & \makecell[r]{\textbf{0.720} \\ $\pm$0.0022} & \makecell[r]{\textbf{0.819} \\ $\pm$0.0029} & \makecell[r]{\underline{0.890} \\ $\pm$0.0024} & \makecell[r]{\underline{0.799} \\ $\pm$0.0005} & \makecell[r]{\underline{0.714} \\ $\pm$0.0009} \\
SEQ-RM & \makecell[r]{\textbf{0.762} \\ $\pm$0.0014} & \makecell[r]{\textbf{0.806} \\ $\pm$0.0031} & \makecell[r]{\textbf{0.721} \\ $\pm$0.0024} & \makecell[r]{\textbf{0.818} \\ $\pm$0.0042} & \makecell[r]{\textbf{0.896} \\ $\pm$0.0014} & \makecell[r]{\textbf{0.801} \\ $\pm$0.0015} & \makecell[r]{0.707 \\ $\pm$0.0015} \\
\bottomrule
\end{tabular}%

\caption{Average grounding and value system accuracies (\PKU). The accuracies for each PKU value are estimated by the scores given by its most related Armo-RM attribute (detailed in supplementary material).}
\label{tab:group-summary2}
\end{table*}

\begin{table*}[t]
\centering
\begin{tabular}{lrrrr|rrrrr}
\toprule
Method & {HF} & {HO} & {TF} & {IF}  & {PF} & {OB} & {CL} &{IR} & {SF} \\
\midrule
VSL-RM& \makecell[r]{0.096 \\ $\pm$0.002} & \makecell[r]{0.347 \\ $\pm$0.008} & \makecell[r]{0.408 \\ $\pm$0.007} & \makecell[r]{0.147 \\ $\pm$0.003} & \makecell[r]{0.223\\$\pm$0.003} & \makecell[r]{0.125 \\ $\pm$0.005} & \makecell[r]{0.368\\$\pm$0.005} & \makecell[r]{0.203\\$\pm$0.003} & \makecell[r]{0.082\\$\pm$0.002} \\
SEQ-RM& \makecell[r]{0.050 \\ $\pm$0.009} & \makecell[r]{0.477 \\ $\pm$0.010} & \makecell[r]{0.316 \\ $\pm$0.006} & \makecell[r]{0.158 \\ $\pm$0.006} & \makecell[r]{0.265 \\ $\pm$0.003} & \makecell[r]{0.110 \\ $\pm$0.006} & \makecell[r]{0.410\\$\pm$0.004} & \makecell[r]{0.168\\$\pm$0.001} & \makecell[r]{0.047\\$\pm$0.001} \\
\bottomrule
\end{tabular}
\caption{Average and standard deviation of value system weights of VSL-RM and SEQ-RM in \Ultra (left) and \PKU (right). HF = Helpfulness, HO = Honesty, TF = Truthfulness, IF = Instruction Following, PF = Prompt Following, OB = Objectivity, CL = Clarity, IR = Information Richness, SF = Safety.}
\label{tab:group-weights}
\end{table*}

Table~\ref{tab:group-summary} (for \Ultra) and Table~\ref{tab:group-summary2} (for \PKU) present the performance of the models on the \textbf{test datasets}. We omit the presentation of grounding accuracies in the case of BT-RM, as the model predictions in value alignment are latent constructs used to minimize value system discordance.

All methods achieve comparable VSA across datasets, with BT-RM consistently performing best and Armo-RM exhibiting the weakest results. In the \PKU dataset, this could be expected, because Armo-RM was not trained on it. However, Armo-RM was trained on \Ultra, and we expected that Armo-RM’s gating mechanism should, in principle, outperform our linear value system representation in predicting overall preferences. This inaccuracy might be due to the fact that its training shared focus with other 5 datasets with possibly divergent overall preference criteria. As observed in Figure~\ref{fig:curves}, BT-RM slightly outperforms VSL-RM regarding VSA in \Ultra\ and more substantially in \PKU, suggesting that, in the latter, value systems may not easily be approximated as linear functions of the dataset value alignment scores. This indicates that more expressive value system models will be necessary to improve performance.

Comparing VSL-RM with SEQ-RM, the former consistently achieves marginally higher VSA, yet with statistically significant differences (one-sided Welch t-tests: $p=0.0004$ in \PKU\ and $p<0.0001$ in \Ultra). This highlights a limitation of sequential optimization: although multiple grounding functions can achieve similar AGA, some are more suitable for linear aggregation when minimizing value system discordance, as have been shown~\cite{andresEcai2025}. The relatively small gap here is likely due to the regularization imposed on the grounding function implementation, given by the use of both $\ell_2$ regularization and the centering parameter $r$.

Regarding grounding accuracy, VSL-RM and SEQ-RM achieve comparable results, consistent with the training curves (differences are not statistically significant\footnote{In supplementary material, we performed one-sided Welch t-tests to compare the metrics achieved by both methods.}). Notably, Armo-RM outperforms all methods in \Ultra, which may be attributed to its per-value training and the use of a MSE loss over quantitative annotations. In contrast, our methods are restricted to less informative qualitative labels, and the use of cross-entropy losses that are prone to overfitting. This limitation could be mitigated in future work, e.g. by label smoothing~\cite{Zhu2024IterativeDataSmoothingForRewardLearning}.

Lastly, in Table~\ref{tab:group-weights}, we analyze the value system weights predicted by VSL-RM and SEQ-RM\footnote{We omit the weights learned by BT-RM, as they are not based on valid value interpretations.}. Both models learned broadly similar value system representations, assigning comparable importance to the values. This indicates that the simultaneous learning procedure in VSL-RM does not substantially distort the learned value system, compared to a grounding implemented by a reward vector learned without a bias to maximize VSA. Notably, however, VSL-RM yields less “radical” weight distributions than SEQ-RM: it maintains a minimum weight of approximately $0.082$ across all values, whereas SEQ-RM assigns some weights below $0.05$. We consider that learning a value system model for a whole dataset that does not neglect any value achieves a better alignment with the hypotheses of the dataset creators, as they would expect that all values have a meaningful contribution towards explaining the overall preferences based on their proposed values. Our learned grounding functions and value system weights with VSL-RM, are, in combination, qualitatively better in this regard. 

As a conclusion, we have shown the advantages of an approach that defends an ``intentional'' interpretation of value alignment, by learning not only accurate approximations of value groundings but also ones which facilitate their combination through value systems as a means for an accurate estimation of value-based preferences. This is a fundamental advantage of our approach that is neglected by most existing methods that learn reward components and user preferences separately or sequentially, an approach that often leads to reward hacking in the fine-tuning phase~\cite{zhou-etal-2024-beyondMODPO}.

\section{Conclusions and Future Work}\label{sec:conclusions}

In this paper, we have put forward a principled approach to value system learning for generative AI applications. Specifically, we have introduced a learning algorithm that, based on preference data, simultaneously learns representations of the value alignment (groundings) of prompt-response pairs (training a multi-objective reward model), and of their alignment with a certain value system (approximated by value importance weights). We have tested our solution against a single-objective reward learning baseline, a sequential 
variant of our algorithm, and a state-of-the art related algorithm.  We have shown that our algorithm prioritizes learning accurate value grounding models to serve as a correct basis for value system estimation, and thus fosters interpretability, with minor losses in capturing overall value-based preferences (value systems). 

Despite these positive results, our work also presents some limitations. First, even though our approach allows for explaining outcomes in terms of grounding functions and value system weights, interpretability is still limited by the fact that 
our algorithm performs an implicit aggregation of preferences regarding values and value systems of possibly divergent opinions, potentially obscuring minorities. Moreover, our current value system representation uses static linear weights, which in certain settings may put limits to the accuracy of modelling real value-based preferences.

In future work, we will address these potential shortcomings without sacrificing interpretability. For this purpose, we will introduce context-dependency into the value system model. A promising option is to dynamically switch between different linear value systems based on contextual conditions. Another major line of work will address methods to learn diverse value systems for heterogenous societies in datasets with known annotation sources. Alternative datasets that we will explore for either task are OpenAssistant~\cite{kopf2023OPENASSISTANTDATASETAlignment} and PRISM~\cite{Kirk2024TheModelsPRISM}. Still, the former involves tree-like conversations, and thus requires an extension of our present model which is geared towards single prompt-response preferences. The latter appears particularly suitable for studying the alignment of our approach with sociological findings.

\section{Acknowledgments}

This work has been supported by grant COSASS: PID2021-123673OB-C32 funded by MCIN/AEI/10.13039/501100011033 and by “ERDF A way of making Europe”, and by project grant EVASAI: PID2024-158227NB-C32 funded by MICIU/AEI/10.13039/501100011033/FEDER, UE. Andr{\'e}s Holgado-S{\'a}nchez has received funding by grant ``Contratos Predoctorales de Personal Investigador en Formaci{\'o}n en Departamentos de la Universidad Rey Juan Carlos (C1 PREDOC 2025)'', funded by Universidad Rey Juan Carlos. Special thanks to the faculty ETSII at Rey Juan Carlos University for providing the computational resources used in the experiments.

\bibliography{aaai2026}


\appendix
\section{Appendix}

\subsection{Source Code}

Source code is available at the following Github repository and branch: \url{https://github.com/andresh26-uam/ValueLearningInGenAI/tree/AIES}.

\subsection{Theoretical considerations}\label{appendix:proofs}
This section provides the theoretical justification and key properties of the solution adopted in our problem formulation, recalled in Problem~\ref{eq:bilevel_loss2}. The scope of the following is to show the soundness of the basic problem formulation, and provide justification for the decisions on our final algorithm. Our algorithm adds some details such as the normalization of the multipliers, that we omit for now for clarity.

\begin{align}\label{eq:bilevel_loss2}
(\omega^*,\theta^*) \in \argmin_{\omega,\theta'} 
&\ \mathcal{L}^{\omega,\theta'}_{\VS}\left(\DS\right) \nonumber\\
\text{subject to} 
\quad \theta' 
&\in \{\theta \mid \mathcal{L}^{\theta}_{v_i}\left(\DS\right)
 \leq \mathcal{L}^*_{v_i}\} \nonumber\\
&\hspace{-0.9em}\text{(for } i\in\{1,\dots,m\}\text{)}
\end{align}

We also recall the theoretical Lagrangian in Eq.~\ref{eq:lagrangian2}, which employs the coefficient vector $\pmb{\lambda} = (\lambda_1,\dots,\lambda_m)$ satisfying $\pmb{\lambda} > \lambda_0$ and $\sum_{i=1}^m \lambda_i = m+1$. In addition, we explicitly incorporate the previously mentioned $\ell_2$ regularization terms.  

\begin{align}\label{eq:lagrangian2}
    \pmb{\mathcal{L}}^{\omega,\theta}_{\pmb{\lambda}}(\DS) =&\mathcal{L}_{\VS}(\DS) + \gamma_{\omega}\abs{\abs{\omega}}^2\nonumber\\ & +\sum_{i=1}^m\lambda_i  \left(\mathcal{L}^{\theta}_{v_i}(\DS) -\mathcal{L}^*_{v_i}\right) \nonumber\\&+ \gamma_{\theta}\abs{\abs{\theta}}^2 +\gamma_{\pmb\lambda}\abs{\abs{\pmb\lambda}}^2 
\end{align}
\begin{proposition}[Strong duality] 

Let $\{\mathcal{L}^{**}_{v_i}\}_{i=1}^m$ denote the exact attainable grounding discordance loss values over the admissible parameter space $\Theta$, i.e. $\mathcal{L}^{**}_{v_i}=\min_{\theta\in\Theta} \mathcal{L}^{\theta}_{v_i}\left(\DS\right)$.

Assume that the parameter space $\Theta$ is dense and that $\Rv^\theta$ is affine in $\theta$. Further assume that the loss targets $\{\loss_{v_i}^*\}_{i=1}^m$ in Problem~\ref{eq:bilevel_loss2} constitute a \emph{strict} approximation of the ideal loss values, i.e.,
\[
\mathcal{L}^*_{v_i} < \mathcal{L}^{**}_{v_i}, \qquad \forall i \in \{1,\dots,m\}.
\] 

Then, strong duality holds: the solutions of the min--max problem~\eqref{eq:DualCondition1} and the max--min (dual) problem~\eqref{eq:DualCondition2} coincide.

\begin{align}
    \min_{\theta,\omega}\left[ \max_{\pmb\lambda>0}\left[\pmb{\mathcal{L}}^{\omega,\theta}_{\pmb{\lambda}}(\DS)\right]\right]\label{eq:DualCondition1}\\= \max_{\pmb\lambda>0}\left[\min_{\theta,\omega}\ \left[\pmb{\mathcal{L}}^{\omega,\theta}_{\pmb{\lambda}}(\DS)\right]\right]\label{eq:DualCondition2}
\end{align}

\end{proposition}
\textit{Proof}. We observe that the loss functions are convex with respect to the difference in predicted rewards. Indeed, the log-sigmoid function is \emph{strictly concave}, and the cross-entropy loss is a negative weighted combination of concave functions; therefore, their composition is convex.

Since the reward vector $\Rv^\theta$ is assumed to be affine in $\theta$, the grounding discordance loss functions are, in particular, convex with respect to $\theta$\footnote{For deep networks, strong duality might not hold. An augmented Lagrangian formulation~\cite{deng2025augmentedlagrangianmethodsoverview} could be used to overcome this issue, at the expense of solution complexity.}. The value system discordance loss is also convex with respect to its parameters, since the reward differences on which it is computed are linear combinations of affine functions. As a consequence, Problem~\ref{eq:bilevel_loss2} is \emph{convex}~\cite{boyd2004convex}. 

Given that we assumed that the parameter space $\Theta$ is dense and the loss functions are, of course, continuous, we can consider the existence of arbitrary points $\theta$ for which $\mathcal{L}^\theta_{v_i} < \mathcal{L}^*_{v_i}$ given any $\mathcal{L}^*_{v_i} > \mathcal{L}^{**}_{v_i}$. Thus, the Slater condition holds, i.e. $\exists\theta\in\textbf{relint}(\Theta) \text{ s.t. } \loss_{v_i}^\theta(\DS) < \mathcal{L}^{*}_{v_i}$. Then, applying what is explained in Section 5.2.3 from~\cite{boyd2004convex}, strong duality holds for any fixed combination of loss targets $\{\mathcal{L}^{*}_{v_i}\}_{i=1}^m$.  $\square$

The preceding result motivates the use of dual ascent methods to solve the dual problem~\cite{boyd2011ADMMreferenceOfDualAscent}, whose solution coincides with that of the primal (Problem~\ref{eq:bilevel_loss2}). However, such methods require exact minimization of the Lagrangian with respect to $\theta$ and $\omega$ for a fixed $\pmb{\lambda}$ at each iteration. We instead approximate the minimization step using gradient descent updates, resulting in an algorithm commonly referred to in game theory as gradient descent--ascent (GDA)~\cite{zamani2022convergencerateanalysisgradient}.

To further justify our approach, we first observe that the $\ell_2$-regularized Lagrangian is \emph{strongly} convex with respect to $\theta$, $\omega$, and $\pmb{\lambda}$. This follows from the fact that $\ell_2$ regularization is itself strongly convex and, when added to convex loss terms, yields a strongly convex objective~\cite{zhou2018fencheldualitystrongconvexity}. Secondly, the logistic loss is $L$-smooth~\cite{fallahSGDAGlobalconvergencewithSMOOTHSTROGCONVEX}, since its second derivative is bounded (see Theorem~5.12 in~\cite{beck2017theorem512forlsmoothIfSecondDerivIsLessThanL}). Then, the Lagrangian is also $L'$-smooth, given that it is a linear combination of smooth functions. If gradient estimates computed on mini-batches are unbiased, i.e., they preserve the expected gradients of the full dataset, then not only gradient descent--ascent (GDA)~\cite{zamani2022convergencerateanalysisgradient}, but also its stochastic counterpart (SGDA), with appropriately chosen step sizes, will converge to an optimal solution~\cite{fallahSGDAGlobalconvergencewithSMOOTHSTROGCONVEX}. Our algorithm can be interpreted as a modified SGDA scheme in which the optimization problem is updated at fixed intervals by updating the loss target estimates $\{\loss^*_{v_i}\}_{i=1}^m$. We argue that this modification does not materially affect convergence, provided that the multipliers are initialized and updated so that the loss targets remain strictly decreasing toward the ideal optimal values $\{\loss^{**}_{v_i}\}_{i=1}^m$ (which can be achieved by experimentally choosing both an adequate $\ell_2$ regularization of multipliers, $\gamma_\lambda$, and a sufficient learning rate $\alpha_\lambda$).

Our final algorithm has a slightly different Lagrangian expression that adds the normalization of the multipliers and other changes. These changes do not invalidate the convexity requirements from before: they are only introduced to regulate the practical implementation of the training algorithm. Namely, this makes learning step sizes similar across iterations and effectively avoid the multipliers from reaching $\lambda_0>0$, which would automatically make impossible to continually decrease the loss target estimates $\{\loss^*_{v_i}\}_{i=1}^m$. Of course, to completely prove these and the previous arguments, we still rely on the experimental validation shown in the main paper.


\subsection{Implementation details}\label{appendix:implementationdetails}

\subsubsection{Transforming quantitative datasets into qualitative ones.}
The considered datasets provide \emph{quantitative} annotations, namely a score $S(\tau)$ that measures the degree of alignment of each generation $\tau$ with a given value or value system. However, both training and evaluation in our framework require only \emph{qualitative} supervision in the form of pairwise preferences\footnote{Armo-RM was trained directly on quantitative annotations and therefore benefits from stronger supervision; we discuss this advantage explicitly in the evaluation section.}. 

A qualitative label encoding the preference between a pair of generations $(\tau,\tau')$ takes values $y \in \{1, 0.5, 0\}$, corresponding to $\tau \succ \tau'$, $\tau \simeq \tau'$, and $\tau \prec \tau'$, respectively. These labels are derived directly from the available scores by comparing their differences:
\[
y =
\begin{cases}
1   & \text{if } S(\tau) > S(\tau'), \\
0   & \text{if } S(\tau) < S(\tau'), \\
0.5 & \text{otherwise}.
\end{cases}
\]

\subsubsection{Qualitative interpretation of model outputs.}
Although the datasets consist of qualitative preferences, our models are \emph{quantitative}: they assign a scalar reward $R$ to each generation and value or value system to estimate the preference data. During learning, we employ the already mentioned Bradley-Terry probability measure of the preference between two responses, $p(\tau \succ \tau'\mid R)$. To obtain predicted qualitative labels, one could in principle apply the same rule as above by replacing $S$ with $R$, or just estimating $y$ with $p(\tau \succ \tau'\mid R)$. However, due to the intrinsic imprecision of neural network outputs, exact equality cases $R(\tau)=R(\tau')$ (or $p(\tau \succ \tau'\mid R)=0.5$) are unlikely to occur in practice.

To account for this, we introduce a tolerance threshold $\epsilon>0$. Whenever the absolute reward difference is smaller than $\epsilon$, the two generations are considered indistinguishable and the predicted label is set to $\hat{y}=0.5$. Only when the difference exceeds this threshold do we predict a strict preference:

$$
\hat{y}_\epsilon(\tau,\tau' \mid R)=\begin{cases}
    1 \text{ if } R(\tau)>R(\tau') +\epsilon\\
    0 \text{ if } R(\tau)<R(\tau')-\epsilon\\
    0.5 \text{ else  } (\abs{R(\tau) -R(\tau')}\leq\epsilon)
\end{cases}
$$

This $\epsilon$ introduces a counterpart threshold in the BT probability space, given by $p_\epsilon = \sigma(\epsilon)-0.5$. This way, we can alternatively estimate the model output preference label $\hat{y}_\epsilon(\tau,\tau' \mid R)$ using the BT model $p(\tau\succ \tau'\mid R)$, as follows: 

$$
\hat{y}_\epsilon(\tau,\tau' \mid R)=\begin{cases}
    1 \text{ if } p(\tau\succ \tau'\mid R) -\sigma(\epsilon) > 0.5\\
    0 \text{ if } p(\tau\succ \tau'\mid R) +\sigma(\epsilon) < 0.5\\
    0.5 \text{ else  } \left(|p(\tau\succ \tau'\mid R) - 0.5| < \sigma(\epsilon)\right) 
\end{cases}
$$

We propose to select the threshold $\epsilon$ by analyzing the label granularity of the available quantitative datasets. Specifically, $\epsilon$ is chosen as half of the minimum non-zero difference between distinct score values. With this choice, adding or subtracting any perturbation $\epsilon' < \epsilon$ to a score does not alter the induced qualitative preference after discretization. For instance, when the dataset uses an integer-valued scoring scheme, we set $\epsilon = 0.5$. In this case, any score difference smaller than $\epsilon$ does not affect the resulting qualitative preference comparison between any pair of generations. 

This criterion was adopted throughout our experiments. Concretely, for \Ultra, we set $\epsilon = 0.25$, as the overall value-alignment scores increase in increments of $0.5$ over the range $[0,10]$. For \PKU, we set $\epsilon = 0.5$, since all annotations are provided on a Likert (1--5) scale.


\subsubsection{Chosen $\epsilon$ for calculating metrics in Table~\ref{tab:group-summary}.}  As explained above, we also used $\epsilon=0.25$ in the \Ultra dataset and $\epsilon=0.5$ in the \PKU dataset for calculating the model predictions of VSL-RM and the newly-trained baselines. Regarding \textsc{Armo-RM}, which was neither trained nor evaluated using the proposed qualitative framework, we selected the values of $\epsilon$ that yielded maximum average grounding accuracy (AGA) and value system accuracy (VSA) over a predefined set of feasible values:
\[
\epsilon \in \{0, 0.001, 0.01, 0.02, 0.05, 0.1, 0.25, 0.5\}.
\]
Empirically, we found that setting $\epsilon = 0.1$ to estimate grounding accuracies and $\epsilon = 0.01$ for value system accuracy resulted in the best overall performance in \Ultra. Regarding \PKU, these thresholds were set at $\epsilon = 0.1$ and $\epsilon=0.0$, respectively. The choice of $\epsilon$ was observed to have a substantial impact on accuracy: for instance, in \Ultra, using the same threshold as in our proposed algorithms ($\epsilon = 0.25$) led to a test AGA of only $0.7672$, while VSA dropped to $0.586$.

\subsubsection{Effects on the training process}
The introduction of the margin $\epsilon$ alters the relationship between minimizing the loss in Eq.~\ref{eq:crossentropy} and minimizing discordances. In particular, a non-zero margin implies that small differences in alignment scores should not induce strict preferences (labels $0$ or $1$).

To make the loss function explicitly aware of this tolerance—i.e., that the difference between alignment scores must exceed $\epsilon$ in order to predict labels $1$ or $0$—we modify the computation of the predicted preference probability $p(\tau \succ \tau' \mid R)$ in Eq.~\ref{eq:crossentropy} by introducing a label-dependent margin, following an approach similar to that adopted in \textsc{UltraRM}~\cite{cui2023ultrafeedback}:

\[
p(\tau \succ \tau' \mid R) \underset{\text{Eq.~\eqref{eq:crossentropy}}}{\underset{\triangle}{=}}
\sigma\!\left( R(\tau) - R(\tau') - m(y) \right),
\]
where $\sigma(\cdot)$ denotes the logistic sigmoid and the margin term $m(y)$ is defined as
\[
m(y) =
\begin{cases}
+\epsilon & \text{if } y = 1, \\
-\epsilon & \text{if } y = 0, \\
0         & \text{if } y = 0.5.
\end{cases}
\]

\subsubsection{Precise calculation of discordance.}
The theoretical definition of discordance introduced in Section~\ref{sec:vslearning} do not explicitly specify how to handle the indifference label $y=0.5$, corresponding to $\tau \simeq \tau'$. To account for this case in a manner consistent with Eq.~\ref{eq:discordance}, we compute the discordance of a model $R$ over a collection of labelled generation pairs $\text{XS} \subset \mathcal{T} \times \mathcal{T} \times \{0, 0.5, 1\},$
using the following empirical formulation:
\begin{equation}\label{eq:discordance_data}
    d_{\text{XS}}\left(R\right) = \frac{1}{\abs{\text{XS}}}\sum_{\substack{(\tau,\tau',y)\in \text{XS}}}\abs{y-\hat{y}_\epsilon(\tau,\tau'\mid R)}
\end{equation}  

This definition assigns, for each element of $\text{XS}$, a discordance of $1$ when $y=0$ and $\hat{y}_\epsilon=1$ (or vice versa); a discordance of $0$ when the label and its corresponding prediction coincide; and a discordance of $0.5$ when $y=0.5$ but $\hat{y}_\epsilon \in \{0,1\}$ (or vice versa). The latter case follows from the interpretation of the indifference relation $\tau \simeq \tau'$ ($y=0.5$) as the simultaneous satisfaction of $\tau \preccurlyeq \tau'$ and $\tau \succcurlyeq \tau'$. Consequently, predicting a strict inequality correctly satisfies one of the two conditions and is therefore penalized by half a unit ($0.5$) rather than a full discordance of $1$.


We must note a special case in the calculation of discordances with Armo-RM. First, one of the six datasets used during the training of Armo-RM is the \Ultra dataset. Accordingly, Armo-RM predicts alignment scores for the values in \Ultra (along with 13 additional values), which we use to estimate grounding discordances in this dataset. However, Armo-RM does not provide predictions for the values in the \PKU dataset; so we predicted the PKU values with the most related one we found among the Armo-RM's attributes (see the last subsection in Appendix~\ref{appendix:experimentdetails} for details). Value system discordance/accuracy is derived from the model’s overall score, which aggregates performance across all 19 values. This metric can be computed for both datasets, as Armo-RM is designed for general-purpose reward benchmarking.

\subsection{Experiment details}\label{appendix:experimentdetails}

\subsubsection{Hardware and Wall-clock times}
All experiments were carried out in two different GPU clusters. One used an AMD EPYC 9654 96-Core Processor with both NVIDIA L40S (48GB) and H100 (96GB). The other uses 2 CPU AMD 7742 64C/128T (256 threads total) and NVIDIA A100 GPUS (80GB). For each training run we employed a single GPU (of any type). Of course, we used the same hyperparameter and library settings across machines. 

Per trained model, we show the approximate wall-clock times in Table~\ref{tab:wallclock} (using the L40S graphics cards). Times may vary due to cluster congestion. Runs executed with the second cluster with A100 graphics cards were $\sim$40\% slower.

\begin{table}[t]
    \centering
    \begin{tabular}{l|rr}
    \toprule
         Method & Train. Steps & Wall-clock time\\
         \midrule
         VSL-RM (PKU) & 23000 & 6h 35m 14 $\pm$55s\\ 
         BT-RM (PKU) & 23000 & 6h 35m 8s $\pm$2m \\
         SEQ-RM (PKU) & 34500 & 9h 33m 43s $\pm$1m\\ 
         \midrule
         VSL-RM (UF) & 26450 & 12h 1m 37s $\pm$9m\\ 
         BT-RM (UF) & 26450 & 11h 59m 30s $\pm$58s \\
         SEQ-RM (UF) & 39675 & 15h 7m 29s $\pm$3m\\ 
    \bottomrule
    \end{tabular}
    \caption{Wall-clock times per method and dataset across the 4 seeds (\textbf{UF}: \Ultra dataset, \textbf{PKU}: \PKU dataset).}
    \label{tab:wallclock}
\end{table}

The wall clock times apparently indicate that our novel algorithm (VSL-RM) has no significant overhead in training time with regard the baselines. However, we deduce the fact that all baselines and VSL-RM share the same codebase affects these results. 

Interestingly, yet no surprisingly, the first learning phase in SEQ-RM was about $\sim$17\% faster than the full training of BT-RM's or VSL-RM's, despite these training settings all required the same number of iterations. We understand this ocurred because SEQ-RM avoided calculating the value system linear layer, which for backpropagation is a expensive process due to the softmax weight calculation. 
 
As a final note, we increased training epochs much beyond the necessary, to check convergence and possible overfitting problems. Reasonably good models could be learned with less than 20\% of the training epochs utilized in the experiments.

\subsubsection{Dataset preprocessing and splitting.}

The datasets are divided into training, validation, and test splits. Validation splits are used to monitor training and select hyperparameters, while test splits are used exclusively for performance evaluation and are never observed during training.


\begin{table}[t]
    \centering
    \begin{tabular}{p{0.49\columnwidth}|p{0.39\columnwidth}}
        Attribute in Armo-RM & Value in \PKU \\
        \toprule
        \textit{Instruction Following} from \Ultra \cite{cui2023ultrafeedback} & \textit{Prompt Following}\\
        \midrule
        \textit{Truthfulness} from \Ultra \cite{cui2023ultrafeedback} & \textit{Objectivity}\\
        \midrule
        \textit{Coherence} from \textsc{Helpsteer} \cite{helpSteerNvidiaWang2024ccby40} & \textit{Clarity}\\
        \midrule
        \textit{Complexity} from \textsc{Helpsteer} \cite{helpSteerNvidiaWang2024ccby40} & \textit{Information} Richness\\
        \midrule
        \textit{Safety} from \textsc{Beavertails} \cite{ji2024safetydataset}  & \textit{Safety}
    \end{tabular}
    \caption{Mapping from the attibutes predicted by the multi-objective Armo-RM model to the values in the PKU dataset. We used the attributes in the first column to predict the alignment with the respective values in the \PKU dataset.}
    \label{tab:mappingpkuarmo}
\end{table}

\begin{itemize}
    \item \textbf{UltraFeedback.} We use the version available in HuggingFace\footnote{\Ultra paper~\cite{cui2023ultrafeedback}, source: \url{https://huggingface.co/datasets/openbmb/UltraFeedback}. We observed the snapshot from May 11, 2026.}. The dataset contains 383{,}796 feedback annotations corresponding to 64{,}000 prompts. For each prompt, approximately four responses are provided, each annotated across four goals (that we treat as values). We reorganize the data by treating each prompt--response pair as an independent instance, yielding 383{,}796 instances in the format of our dataset $\DS$. Some value labels are missing for specific instances, so they were excluded from the computation of the metrics.
    
    \Ultra does not provide official training, validation, and test splits. We therefore construct them as follows: 10\% of the data (38{,}380 instances) is reserved for testing. From the remaining data, 2\% (6{,}909 instances) is used for validation, and the remaining 98\% (338{,}507 instances) is used for training.

    \item \textbf{PKU-Align-Anything (text-to-text).} We use the text-to-text subset of the PKU dataset\footnote{\PKU paper~\cite{ji2024alignanythingtrainingallmodality}, Hugging Face source: \url{https://huggingface.co/datasets/PKU-Alignment/align-anything/viewer/text-to-text/}. We observed the snapshot from May 11, 2026.}, which contains 31{,}430 annotated prompt--response pairs. The annotations include  quantitative value alignment scores, qualitative value alignment preferences, and an overall qualitative preference label indicating the preferred response. We use the overall preference label as the value-system alignment label $y^j_V$ (with a fixed $j$, as annotator identities are not provided). We use the quantitative alignment scores to generate respective value alignment preference labels between responses, as we found the qualitative annotations to be inconsistent with the quantitative ones and, wrongly, tend to align with the overall preference label by taking into account other criteria apart from the values.  

The original dataset provides a training split (``train'') and a validation split (``val''). We treat ``train'' as a combined training--validation pool and ``val'' as the test dataset. From the combined pool, we select the first 29{,}430 instances for training and use the remaining 1{,}000 instances for validation.


\end{itemize}

\subsubsection{Approximating PKU values with Armo-RM attributes.}
Because Armo-RM was not trained on PKU, we decided to predict the alignment with the values of the PKU dataset with the scores given to the most similar attributes in the Armo-RM reward model. These proxy scores are used as a measure of grounding accuracy in Table~\ref{tab:group-summary2}. The chosen correspondence between PKU values and Armo-RM attributes is given in Table~\ref{tab:mappingpkuarmo}.

Three of the mappings can be considered as a perfect equivalence, yet with care, as the distributions of the data are different. First, \textit{prompt following} is certainly similar if not equal to \textit{instruction following} from \Ultra. Second, \textit{objectivity} is included in the defintion of \textit{truthfulness} from \Ultra, as it includes being ``faithful to factual knowledge''. Third, \textit{Safety} is clearly already predicted by Armo-RM via the \textit{safety} parallel in the \textsc{Beavertails} dataset \cite{ji2024safetydataset}.

\begin{table}[t]
    \centering
    \begin{tabular}{lp{6.2cm}}
        \toprule
        \textbf{Symbol} & \textbf{Description} \\
        \midrule
        $N$ & Training epochs: Number of training passes (stepping through the necessary batches) across full training dataset.\\
        $b$ & Batch size. \\
        $gr_{acc}$ & Gradient accumulation steps \\
        $\alpha_{\theta}, \alpha_{\omega}$ & Learning rates for model parameters ($\theta$, $\omega$).\\
        $\alpha_{\lambda}$ & Learning rate for Lagrange multipliers.\\
        $\gamma_{\theta},\gamma_{\omega}$ & Weight decay factors for $\ell_2$ regularization of the main model parameters $\theta$ and $\omega$. \\
        $\gamma_\lambda$ & Weight decay factor for the norm of the Lagrange multiplier parameters ($\ell_2$ regularization).\\
        $\eta$ & Update ratio of the ``exponentially weighted minimum'' formula that calculates the target losses $\mathcal{L}^*_\VS$ and $\mathcal{L}^*_V$. \\
        $u$ & Number of passes through step 1-2 before updating the target losses $\mathcal{L}^*_\VS$ and $\mathcal{L}^*_V$.\\
        $r$ & Reward centering coefficient (Eq.~\ref{eq:crossentropy}). \\
        \bottomrule
    \end{tabular}
    \caption{Glossary of hyperparameters used in the experiments.}
    \label{tab:hp-glossary}
\end{table}

\subsubsection{Hyperparameters.}
In Table~\ref{tab:hyperparameters}, we report the full set of hyperparameters used in our experiments for each method. Explanations of individual hyperparameters are provided in the glossary tables (Table~\ref{tab:hp-glossary}). These parameters were selected based on configurations that tended to minimize the average total discordance on the PKU-Align-Anything dataset when using VSL-RM. The values were obtained through a standard Bayesian search, with the number of training epochs restricted to a 10\% of the epochs used in the final experiments.

\begin{table*}[t]
    \centering
    \begin{tabular}{c|rrr}
    \toprule
    Hyperparameter & VSL-RM & BT-RM & SEQ-RM \\
    
    \midrule 
     
    $b$ & 32 & 32 & 32 \\
    $gr_{acc}$ & 4 & 4 & 4 \\ 
    $\alpha_{\theta}$ & 0.0001 & 0.0001  & 0.0001\\
    $\alpha_{\omega}$ & 0.0001 & 0.0001 & 0.001\\
    $\alpha_\lambda$ & 0.05 & - & - \\
    $\gamma_{\theta}$ & 0.003 & 0.003  & 0.003 \\  
    $\gamma_{\omega}$ & 0.003 & 0.003 & 0.003 \\  
    $\gamma_\lambda$ & 0.005 & -  & -\\
    $\eta$ & 0.9 & - & - \\
    $u$ & 20 & - & - \\
    $r$ & 0.01 & 0.01 & 0.01 \\ 
    \midrule
    Epochs ($N$) (UF$\mid$PKU) & $10\mid100$ & $10\mid100$ & $10+5\mid100+50$ \\  
    Training steps ($T$) (UF$\mid$PKU) & $26450\mid23000$ & $26450\mid23000$  & $26450+13225\mid23000+11500$\\ 
    \bottomrule
    \end{tabular}
    \caption{Hyperparameters used in each environment per algorithm. Below, the number of epochs in each dataset (UF: \Ultra, PKU: \PKU. }
    \label{tab:hyperparameters}
\end{table*}

\subsection{Algorithm pseudocode}\label{appendix:algorithm}

In Algorithm~\ref{alg:algorithm} we provide the pseudocode of the described algorithm in Section~\ref{sec:vslearning}.

\begin{algorithm*}[t]
\caption{Value System Learning with Reward Models via Stochastic Gradient Descent-Ascent (\textbf{VSL-RM})}
\label{alg:algorithm}
\begin{algorithmic}[1]
\Statex Let $\DS$ a training dataset with the structure described in the paper. Set a desired $\lambda_0>0$ and number of training epochs $N$. 
\Statex Choose hyperparameters: batch size $b>0$, gradient accumulation steps $gr_{acc}$, learning rates $\alpha_\theta,\alpha_\omega,\alpha_\lambda$, $\ell_2$ regularization weights $\gamma_\theta,\gamma_\omega,\gamma_\lambda$, and exponential weighted minumum update coefficient $\eta < 1$ (e.g. approx $0.9$).

\Statex Initialize trainable parameters $\theta,\omega,\bar{\pmb{\lambda}}\in\mathbb{R}^{m+1}$ and gradient accumulator variables $g_\theta,g_\omega=0$.

\For{$e = 1$ to $N$}
Divide $\DS$ into batches of size $b$: $\text{BS}_k$, $k\in \{1,\dots,K\}$.
    \For{$k=1,\dots,K$}
    \State Calculate and save value system discordance loss $\loss_{\VS}^{\omega,\theta}(\text{BS}_k)$, and grounding discordance losses: , $\loss_{v_i}^{\theta}(\text{BS}_k)$, for each $i\in\{1,\dots,m\}$.
    
    \If{$e=1 \wedge k\leq u$}
    
    \If {$k \bmod u=0$}
    \State $\mathcal{L}_{v_i}^* \gets \sum_{k'=0}^u \mathcal{L}_{v_i}^\theta(\text{BS}_{k'})$\Comment{Initialization of loss targets}
        
    \EndIf
    \State \textbf{continue}
    \EndIf
        
        \State Calculate normalized multipliers:
        \[
        (\lambda_1,\dots,\lambda_m, \lambda_{\VS}) =
        \frac{(m+1)\exp(\bar{\pmb{\lambda}})}
        {\sum_{\bar{\lambda}\in\bar{\pmb{\lambda}}}\exp(\bar{\lambda})}
        (1-\lambda_0) + \lambda_0
        \]
        
        \State \textbf{(1) \underline{Gradient descent}}

\begin{align*}
    \bar{\loss}^{\omega,\theta}_{\bar{\pmb{\lambda}}}(\text{BS}_k) &= \lambda_\VS\mathcal{L}_{\VS}^{\omega,\theta} (\text{BS}_k) + \sum_{i=1}^m\lambda_i  \left(\mathcal{L}^{\theta}_{v_i}(\text{BS}_k)-\mathcal{L}^*_{v_i}\right) + \gamma_\theta\|\theta\|^2 + \gamma_\omega\|\omega\|^2+ \gamma_\lambda\|\bar{\pmb\lambda}\|^2
\end{align*}
        \State Accumulate gradients: $g_\theta\gets g_\theta + \nabla_{\theta}\bar{\loss}^{\omega,\theta}_{\bar{\pmb{\lambda}}}(\text{BS}_k)$;  $g_\omega\gets g_\omega + \nabla_{\omega}\bar{\loss}^{\omega,\theta}_{\bar{\pmb{\lambda}}}(\text{BS}_k)$
        \If{$k$ $\bmod\ gr_{acc} = 0$}
            \State $\theta \gets \theta - \frac{\alpha_{\theta}}{gr_{acc}}\cdot g_\theta$
            \State $\omega \gets \omega - \frac{\alpha_{\omega}}{gr_{acc}}\cdot g_\omega$
            \State Reset gradient accumulators $g_\theta\gets0$, $g_\omega\gets 0$

        \State \textbf{(2) \underline{Gradient ascent on multipliers}}
        
        \[ \bar{\pmb\lambda} \gets \bar{\pmb\lambda} + \alpha_\lambda\nabla_{\bar{\pmb{\lambda}}}\left[
        \sum_{i=1}^m \lambda_i \max\!\left(\sum_{k'=k-gr_{acc}+1}^k\left(\frac{1}{gr_{acc}}{\mathcal{L}^\theta_{v_i}(\text{BS}_{k'})}\right) - \mathcal{L}^*_{v_i},0\right)
        + \gamma_\lambda \|\bar{\pmb{\lambda}}\|^2\right]
        \]
        \EndIf
        \If{$k$ $\bmod\ u = 0$}
            \State \textbf{(3) \underline{Target loss update}}
            \For{$i = 1,\dots,m$}
                \State $\mathcal{L}_{v_i}^* \gets 
                \eta \mathcal{L}^*_{v_i} + (1-\eta)
                \min\!\left(\mathcal{L}^*_{v_i}, \frac{1}{u}\sum_{k'=k-u+1}^k \mathcal{L}_{v_i}^\theta(\text{BS}_{k'})\right)$
            \EndFor
        
        \EndIf

    \EndFor
\EndFor
\end{algorithmic}
\end{algorithm*}

\subsection{Statistical significance tests}
In the main paper, we report several claims regarding the statistical significance of the differences between the test metrics obtained by VSL-RM and SEQ-RM. Here, we provide the detailed results of these statistical tests in Table~\ref{tab:pairwise-testsultra} for the \Ultra\ dataset, and in Table~\ref{tab:pairwise-testspku} for the \PKU\ dataset.

\begin{table*}[t]
    \centering

\begin{tabular}{llllrrl}
\toprule
Method A & Method B & Metric & Test & Statistic $t$ & One-tailed $p$-value & Sig. \\
\midrule
VSL-RM & SEQ-RM & Helpfulness & Welch t-test ($A<B$) & -1.78 & 0.0695 & ns \\
VSL-RM & SEQ-RM & Honesty & Welch t-test ($A<B$) & -0.60 & 0.2857 & ns \\
VSL-RM & SEQ-RM & Truthfulness & Welch t-test ($A<B$) & -2.87 & 0.0142 & * \\
VSL-RM & SEQ-RM & Instruct. Follow. & Welch t-test ($A<B$) & -1.19 & 0.1406 & ns \\
VSL-RM & SEQ-RM & AVG. & Welch t-test ($A<B$) & -4.06 & 0.0046 & ** \\
VSL-RM & SEQ-RM & VSA & Welch t-test ($A>B$) & 13.15 & 0.0000 & *** \\
\bottomrule
\end{tabular}

\caption{Pairwise one-tailed Welch t-tests comparing metrics achieved by VSL-RM and SEQ-RM in the \Ultra dataset: we test $H_1: \mu_{A} > \mu_{B}$ when Method A's sample mean exceeds Method B's, otherwise we test $H_1: \mu_{A} < \mu_{B}$.}\label{tab:pairwise-testsultra}
\end{table*}

\begin{table*}[t]
    \centering

\begin{tabular}{llllrrl}
\toprule
Method A & Method B & Metric & Test & Statistic $t$ & One-tailed $p$-value & Sig. \\
\midrule
VSL-RM & SEQ-RM & Prompt Follow. & Welch t-test ($A>B$) & 1.47 & 0.1050 & ns \\
VSL-RM & SEQ-RM & Object. & Welch t-test ($A<B$) & -3.64 & 0.0055 & ** \\
VSL-RM & SEQ-RM & Clarity & Welch t-test ($A<B$) & -0.38 & 0.3580 & ns \\
VSL-RM & SEQ-RM & Inf. Rich. & Welch t-test ($A>B$) & 0.20 & 0.4263 & ns \\
VSL-RM & SEQ-RM & Safety & Welch t-test ($A<B$) & -4.09 & 0.0053 & ** \\
VSL-RM & SEQ-RM & Avg. Chr. & Welch t-test ($A<B$) & -2.85 & 0.0254 & * \\
VSL-RM & SEQ-RM & Represent. & Welch t-test ($A>B$) & 7.29 & 0.0004 & *** \\
\bottomrule
\end{tabular}

\caption{Pairwise one-tailed Welch t-tests comparing metrics achieved by VSL-RM and SEQ-RM in the \PKU dataset: we test $H_1: \mu_{A} > \mu_{B}$ when Method A's sample mean exceeds Method B's, otherwise we test $H_1: \mu_{A} < \mu_{B}$.}\label{tab:pairwise-testspku}
\end{table*}

\end{document}